\documentclass[letterpaper, 10pt, conference]{ieeeconf}  
\usepackage{amsmath,amsfonts}
\usepackage{algorithmic}
\usepackage{array}
\usepackage{textcomp}
\usepackage{stfloats}
\usepackage{url}
\usepackage{verbatim}
\usepackage{graphicx}
\usepackage{balance}
\usepackage[ruled,vlined,linesnumbered]{algorithm2e}
\usepackage{color} 
\usepackage{bm}
\usepackage{amsthm}  
\usepackage{mathtools}
\usepackage{subfigure}
\usepackage{flushend}

\IEEEoverridecommandlockouts                              

\title{\LARGE \bf
Task and Motion Planning of Dynamic Systems \\ using Hyperproperties for Signal Temporal Logics\vspace{-10pt}
}
\author{Jianing~Zhao,
	Bowen~Ye,
	Xinyi~Yu, Rupak~Majumdar and
	Xiang~Yin
	\thanks{This work was supported by the National Natural Science Foundation of China (62061136004, 62173226, 61803259). 
 }
\thanks{J. Zhao, B. Ye and X. Yin  are with the School of Automation and Intelligent Sensing, Shanghai Jiao Tong University, and the Key Laboratory of System Control and Information Processing, the Ministry of Education of China, Shanghai 200240, China. {\tt  E-mail: \{jnzhao,yebowen1025,yinxiang\}@sjtu.edu.cn}.}
\thanks{
X. Yu is with the Thomas Lord Department of Computer Science, University of Southern California, Los Angeles, CA 90089, USA. 
}
\thanks{
R. Majumdar is with the Max-Planck Institute for Software Systems, Kaiserslautern 67663, Germany. {\tt  E-mail: rupak@mpi-sws.org}. (Corresponding Author: Xiang Yin)
}
}

\def \u{{\mathbf{u}}}
\def \p{{\bm{p}}}
\def \w{{\mathbf{w}}}
\def \x{{\mathbf{x}}}

\def \until{{\mathbf{U}}}
\def \eventually{{\mathbf{F}}}
\def \always{{\mathbf{G}}}

\def \U{\mathcal{U}}
\def \UU{\mathbb{U}}
\def \T{\mathcal{T}}
\def \V{\mathcal{V}}
\def \X{\mathcal{X}}
\def \RR{\mathbb{R}}
\def \NN{\mathbb{N}}
\def \res{\emph{\text{res}}}
\def \true{\emph{\text{True}}}
\def \false{\emph{\text{False}}}
\def \Null{\emph{\text{Null}}}

\newtheorem{myprob}{Problem}
\newtheorem{myrem}{Remark}
\newtheorem{mydef}{Definition}

\newtheorem{mythm}{Theorem}

\begin{document}

\maketitle


\begin{abstract}
We investigate the task and motion planning problem for dynamical systems under signal temporal logic (STL) specifications. Existing works on STL control synthesis mainly focus on generating plans that satisfy properties over a single executed trajectory. In this work, we consider the planning problem for \emph{hyperproperties} evaluated over a set of possible trajectories, which naturally arise in information-flow control problems. Specifically, we study discrete-time dynamical systems and employ the recently developed temporal logic HyperSTL as the new objective for planning. To solve this problem, we propose a novel recursive counterexample-guided synthesis approach capable of effectively handling HyperSTL specifications with multiple alternating quantifiers. The proposed method is not only applicable to planning but also extends to HyperSTL model checking for discrete-time dynamical systems. Finally, we present case studies on security-preserving planning and ambiguity-free planning to demonstrate the effectiveness of the proposed HyperSTL planning framework.
\end{abstract}


\section{Introduction}
Planning and decision-making are fundamental problems in autonomous robotics and \emph{cyber-physical systems} (CPS). 
In recent years, there has been growing interest in task and motion planning  for \emph{high-level specifications} \cite{liu2022secure}. 
To specify objectives for CPS, various temporal logics have been developed, offering expressive and user-friendly tools for formally describing and synthesizing complex tasks. 
Particularly, signal temporal logic (STL) \cite{maler2004monitoring} provides a systematic language for describing complex tasks in continuous dynamical systems with real-valued signals. 
It can express specifications such as ``remain in a region for at least two minutes and then reach another region within three minutes." 
Recently, STL-based planning has been extensively studied and successfully applied in various engineering CPS, including autonomous robots \cite{yu2024online}, power systems \cite{park2020mitigation}, and traffic management \cite{arechiga2019specifying}.

In the context of temporal logic planning for dynamical systems, a central problem is to find a sequence of control inputs such that the system trajectory satisfies a given STL specification. 
To solve the STL task and motion planning problem, different techniques are proposed including the mixed integer programming \cite{raman2014model}, control barrier functions \cite{lindemann2018control} and gradient-based optimizations \cite{mehdipour2019arithmetic}.
When accounting for system uncertainties or disturbances, the planning problem can be integrated into a model predictive control  framework to design feedback controllers that dynamically compensate for real-time perturbations \cite{raman2014model, raman2015reactive,farahani2018shrinking, meng2023signal}. Recent work has extended these methods to multi-agent settings, where the goal is to coordinate multiple trajectories to satisfy a global STL specification \cite{sun2022multi,liu2025controller}.

\begin{figure}
    \centering
    \includegraphics[scale=0.27]{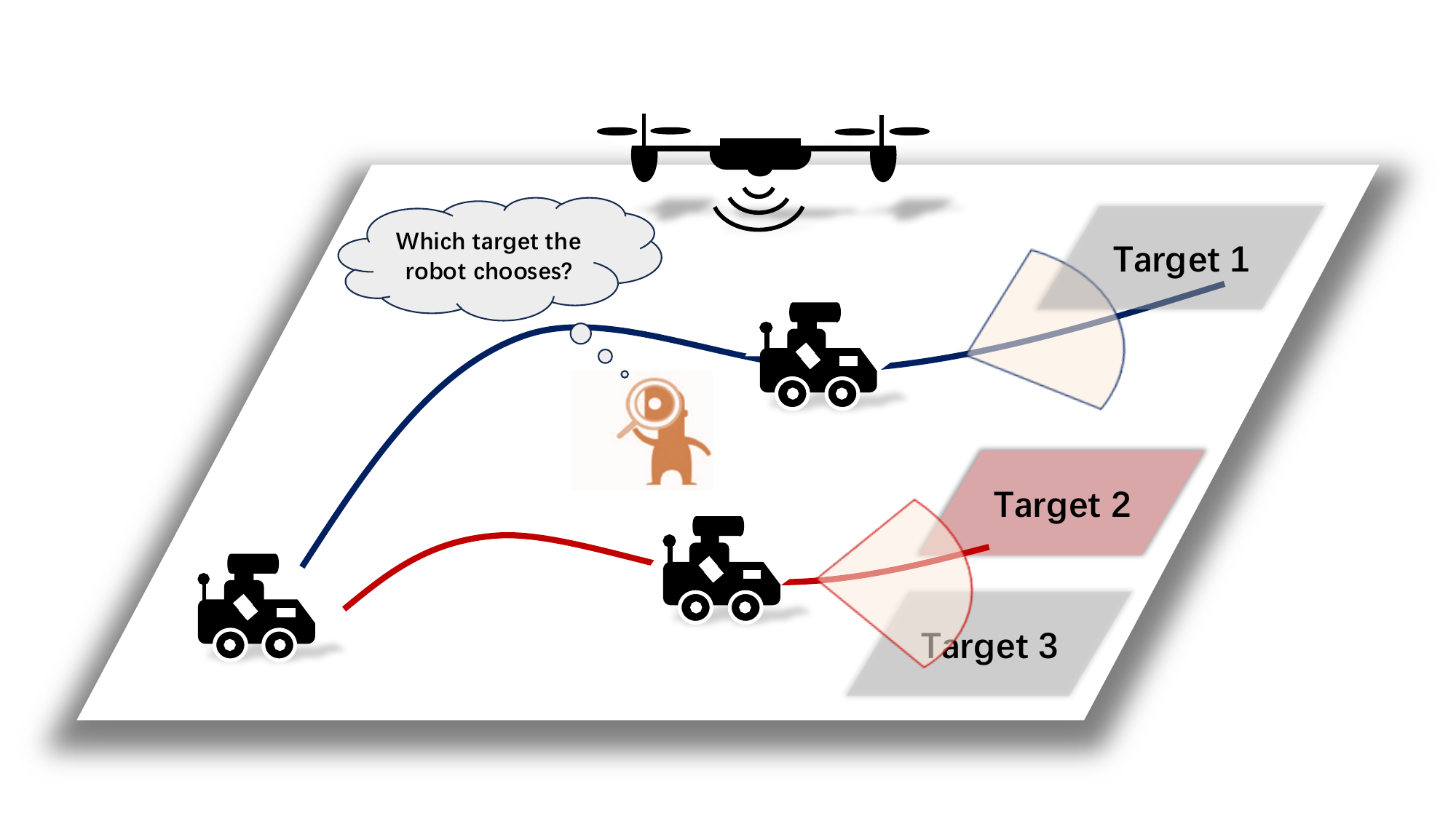}\vspace{-10pt}
    \caption{Motivating Example}
    \label{fig:example}
\end{figure}

However, all the aforementioned works focus solely on ensuring the correctness of one system execution.
In many practical applications, it is also essential to reason about the system's behavior across a set of possible executions, a concept known as \emph{hyperproperties}  \cite{clarkson2010hyperproperties,anand2021formal, liu2022secure, zhao2024unified}.
To illustrate this, we consider a motivating scenario depicted in Figure~\ref{fig:example}. 
Suppose a robot navigates a workspace aiming to reach one of three possible destinations. 
Meanwhile,  a malicious agent attempts to infer the destination of the robot, using the knowledge of 
(i) the robot's dynamics, 
(ii) its planning objective, and (iii) its real-time state trajectory.
From a security perspective, the robot aims to avoid revealing its true destination prematurely.
If the robot optimizes only for task completion, it could take either the blue or red trajectory shown in the figure, where the shaded sectors represent its reachable set within the next two time steps. 
However, choosing the blue trajectory would allow the intruder to conclusively determine that the robot is heading to Target 1 two steps in advance. 
In contrast, the red trajectory leaves the intruder uncertain until arrival as the robot could be moving toward either Target 2 or Target 3. Thus, the red trajectory not only fulfills the task but also preserves security by concealing critical information.


In this paper, we address the motion planning problem for dynamical systems under hyperproperties specified in HyperSTL. 
We focus on a fragment called existential  HyperSTL, which is suitable for trajectory planning.
Particularly, in the context of task and motion planning, 
\cite{wang2020hyperproperties} investigated the problem of synthesizing open-loop plans satisfying HyperLTL specifications. 
This work was later extended in \cite{bonnah2023motion}, where the authors studied planning under HyperTWTL specifications, which is an extension of TWTL that incorporates explicit timing constraints for hyperproperties.
However, both HyperLTL and HyperTWTL are limited to discrete and logical specifications. 
For many CPS applications, we require specifications that can directly reason about continuous system dynamics, such as the nonholonomic constraints of mobile robots. While HyperSTL has been demonstrated as a powerful framework for quantitative evaluation of real-valued signals over multiple traces, the  planning problem for continuous dynamical systems under HyperSTL specifications remains, to the best of our knowledge, an open challenge in the field.
Our approach builds upon  the mixed integer  programming-based optimization for STL trajectory planning, and   counterexample-guided synthesis for STL reactive synthesis.  We integrate these methods into a novel recursive counterexample-guided synthesis framework capable of handling the quantifier alternations over multiple traces inherent in HyperSTL specifications.
The main contributions of this work are threefold. First, in contrast to prior work on hyperproperty synthesis in purely logical settings \cite{wang2020hyperproperties,bonnah2023motion}, our solution can handle continuous systems with real-valued signals. 
Second, unlike existing HyperSTL model checking approaches \cite{nguyen2017hyperproperties,bartocci2023mining} that operate on finite sets of enumerated signals, our method can synthesize plans for discrete-time dynamical systems that generate signals. 
Finally, we demonstrate the practical applicability of our framework through case studies involving information-flow security properties in robot motion planning scenarios.

The remaining parts are organized as follows. 
In Section~\ref{sec:problem}, 
we review some basic preliminaries and formulate the HyperSTL  planning problem. 
The main synthesis algorithms are presented in Section~\ref{sec:algorithm}. 
In Section~\ref{sec:case}, several case studies are presented 
motivated by  information-flow properties for dynamic systems under STL specifications. 
Finally, we conclude the paper in Section~\ref{sec:conclusion}.

\section{Problem Formulation}\label{sec:problem}

\subsection{System Model}
We consider a dynamic system of form\vspace{-3pt}
\begin{equation}\label{eq-dym}
   \Sigma: \  x_{t+1}=f(x_t,u_t),\vspace{-3pt}
\end{equation}
where $x_t\!\in\!\X\subseteq \RR^n$ and $u_t\!\in\!\U\subseteq\RR^m$ 
are the system state and  control input at instant $t$, respectively, 
and $f\!:\!\X\times\U\!\to\!\X$ is the transition function describing the dynamic of the system.
We assume  that the initial state $x_0\in\X$ is given.

Suppose  at time instant $t\!\in\!\NN$, the system $\Sigma$ is in state $x_t\!\in\!\X$. 
Then given a sequence of control inputs $\u_{t:N-1}\!=\!u_tu_{t+1}\cdots u_{N-1}\!\in\!\U^{N-t}$, the solution of the system is $\xi_f(x_t,\u_{t:N-1})\!=\!x_{t+1}\cdots x_N\in\X^{N-t}$ such that $x_{i+1}\!=\!f(x_i,u_i),\forall i\!\geq \!t$.
A \emph{finite} state sequence $x_0x_1\cdots x_N$ is said to be a 
 \emph{trace} of system $\Sigma$ if it is a solution from $x_0$ under some control inputs. 
We denote by $\T_{\Sigma}(x_0)$ and $\T_\Sigma$ the set of all traces generated from the initial state $x_0\in\X_0$ and the set of all traces generated by system $\Sigma$ from any states. 

\subsection{Signal Temporal Logics}
To describe the high-level specifications of system trajectories, we use signal temporal logic (STL) with bounded-time temporal operators, defined by the following syntax:\vspace{-3pt}
\begin{equation}\label{eq-stl}
\phi::=\top\mid\mu_\nu\mid\neg\phi\mid\phi_1\wedge\phi_2\mid\phi_1\until_{[a,b]}\phi_2,\vspace{-3pt}
\end{equation}
where $\top$ is the true predicate, 
$\mu_\nu$ is a predicate whose truth value is determined by the sign of its underlying predicate function $\nu\!:\!\RR^n\!\to\!\RR$, i.e., 
for state $x\in\RR$,  $\mu_\nu$ is true iff $\nu(x)\geq 0$. 
Notations $\neg$ and $\wedge$ are the standard Boolean operators ``negation" and ``conjunction", respectively, which can further induce ``disjunction"  $\phi_1\!\vee\!\phi_2$ and ``implication"  $\phi_1\!\to\!\phi_2$. Notation $\until_{[a,b]}$ is the temporal operator ``until", where $a,b\in\RR_{\geq0}$ are  time instants, and it can also induce  temporal operators ``eventually" and ``always" by $\eventually_{[a,b]}\phi:=\top\until_{[a,b]}\phi$ and $\always_{[a,b]}:=\neg\eventually_{[a,b]}\neg\phi$, respectively.

STL formulae are evaluated on  state sequences. 
We denote by $(\x,t)\models\phi$  that sequence $\x$ satisfies $\phi$ at instant $t$, and we write $\x\models\phi$ whenever $(\x,0)\models\phi$. 
The reader is referred to \cite{maler2004monitoring} for more details on the (Boolean) semantics of STL. 
Particularly, we have  
$(\x,t)\models \mu_{\nu}$    iff   $\nu(x_t)\geq 0$, and
$(\x,t)\models \phi_1\until_{[a,b]}\phi_2$    iff  there exists $t'\in[t+a,t+b]$ such that $(\x,t')\models\phi_2$ and for any $ t''\in[t+a,t']$, we have $(\x,t'')\models\phi_1$. 
In some cases, it is useful to quantitatively evaluate the robustness of STL \cite{donze2010robust}. 
We denote by $\rho^\phi(\x,t)$ the robust value of sequence $\x$  at instant $t$ 
and we have  $(\x,t)\models\phi$ iff $\rho^\phi(\x,t)>0$.

\vspace{-3pt}
\subsection{HyperSTL} 
Let $\V=\{\pi_1,\pi_2,\ldots\}$ be a set of trace variables, where each $\pi_i$ represents an individual trace.
The syntax of HyperSTL is given by \cite{nguyen2017hyperproperties}:
\vspace{-3pt}
\begin{subequations}\label{eq-syntax}
    \begin{align}
        \Phi&::=\exists \pi.\Phi\mid\forall \pi.\Phi\mid\phi\\
    \phi&::=\top\mid\mu_{\nu}^\Theta\mid\neg\phi\mid\phi_1\wedge\phi_2\mid\phi_1\until_{[a,b]}\phi_2\label{eq-stlhyper}\vspace{-6pt}
    \end{align}
\end{subequations}
where $\forall$ and $\exists$ are the universal and existential quantifiers, respectively. 
Note that $\phi$ is essentially an STL formula, 
and the only difference is that 
predicate $\mu_\nu^\Theta$ is parameterized by a   set of trace variables $\Theta\subseteq\V$. Specifically, 
$\mu_\nu^\Theta$ is a predicate whose value is determined by its underlying predicate function $\nu:\RR^{n\times|\Theta|}\to\RR$ such that 
$\mu_\nu^\Theta$ is true iff 
$\nu(x^{\Theta})> 0$, 
where $x^\Theta:=(x^1,\ldots,x^{|\Theta|})\in\RR^{n \times|\Theta|}$ 
is a  state tuple
and each $x^i\!\in\!\X$ denotes the corresponding state of trace $\pi_i$ in $\Theta$. 

Therefore, different from the standard STL,  $\phi$ defined in \eqref{eq-stlhyper} involves multiple trace variables. 
Let  $\x^\Theta\!=\!(\x^1,\x^2,\ldots,\x^{|\Theta|})$ 
be a trace tuple, and we denote $\x^\Theta_t \coloneq  (x_t^1,x_t^2,\ldots,x_t^{|\Theta|})$ as the state tuple at instant $t$.  The semantics of HyperSTL are defined  over a (finite or infinite) set of traces $\T$ and a partial mapping (called \emph{trace assignment}) $\Pi\!:\!\V\!\to\! \T$. 
We use notation $(\Pi,t)\models_\T\Phi$ to denote that a HyperSTL formula $\Phi$ is satisfied by a set of traces $\T$ at time instant $t$. 
The validity judgement of a HyperSTL formula at time instant $t$ is defined  recursive by \cite{nguyen2017hyperproperties}:
\begin{equation}{\small
\begin{tabular}{lcl}
$(\Pi,t)\models_\T \exists\pi.\Phi$  &  iff  & $\exists \xi\in \T:\Pi(\pi)=\xi\wedge\xi\models\Phi$\\
$(\Pi,t)\models_\T \forall\pi.\Phi$  &  iff  & $\forall \xi\in \T:\Pi(\pi)=\xi\wedge\xi\models\Phi$\\
$(\Pi,t)\models_\T \mu_\nu^\Theta$  &  iff  & $\nu(\Pi(\Theta)_t)\geq0$\\
$(\Pi,t)\models_\T \neg\phi$ &  iff  & $(\Pi,t)\not\models_\T\phi$\\
$(\Pi,t)\models_\T \phi_1\wedge\phi_2$  &  iff  & $(\Pi,t)\models_\T\phi_1\wedge(\Pi,t)\models_\T\phi_2$\\
$(\Pi,t)\models_\T \phi_1\until_{[a,b]}\phi_2$  &  iff  & $\exists t'\!\in\![t\!+\!a,t\!+\!b]\!:\!(\Pi,t')\!\models_\T\!\phi_2$\\
&& $\forall t''\!\in\![t+a,t']\!:\!(\Pi,t'')\!\models_\T\!\phi_1$\notag
\end{tabular} }
\end{equation}
We say a set of traces $\T$ satisfies HyperSTL formula $\Phi$, denoted by $\T\models\Phi$, if $(\Pi_{\emptyset},0)\models_\T\Phi$. We say a system $\Sigma$ satisfies $\Phi$, denoted by $\Sigma\models\Phi$, if $\T_\Sigma\models\Phi$.


\subsection{HyperSTL Planning Problem}
 
In the context of  task and motion planning, one  needs to find a specific trace  to execute. We consider  a fragment of HyperSTL called ``\emph{Existential HyperSTL}" ($\exists$-HyperSTL) that starts with the existential quantifier $\exists$ in the form of $\Phi::=\exists\pi.\Phi'$, where $\Phi'$ is an arbitrary HyperSTL formula.

\begin{myprob}
    Given system \eqref{eq-dym} with  $x_0\!\in\!\X$ and   $\exists$-HyperSTL formula $\Phi\!=\!\exists\pi.\Phi'$, check whether  the system $\T_\Sigma(x_0)$ satisfies $\Phi\!=\!\exists\pi.\Phi'$. If so, find the control input sequence $\u\!\in\!\U^N$ such that $\pi\!=\!\xi_f(x_0,\u)$ is an instance satisfying $\Phi$.
\end{myprob} 


\section{Planning from HyperSTL Specifications}\label{sec:algorithm}
In this section, we present our solution approach for the HyperSTL planning problem. 
We begin by examining two special cases: (i) alternation-free formulae and (ii) formulae with alternation depth 1. Building on these two special cases, we then develop a general solution that combines and extends the techniques developed for each case.

\subsection{Case of Alternation-free HyperSTL}

First, we consider the  synthesis problem for the special case of alternation-free HyperSTL of the  following form 
\begin{equation}\label{eq-alternationfree}
    \Phi=\exists\pi_1.\exists\pi_2.\cdots.\exists\pi_n.\phi.
\end{equation}
This fragment can be handled by extending the technique for  standard  STL synthesis. 
Essentially, it is equivalent to consider an augmented system that consists of $n$ copies of the original system 
and then to find a $n$-tuple of  control input sequences $\u_1,\ldots,\u_n\!\in\!\U^N$, where $\u_i\!=\!u_0^iu_1^i\cdots u_{N-1}^i$, 
such that they jointly satisfy the STL specification $\phi$. 

To this end, we first encode the STL constraints by binary variable according to the technique in \cite{raman2014model}. We denote by $\texttt{constr}(\phi)$ the set of all constraints that encode the satisfaction of STL formula $\phi$.

As for the objective function, since we are only interested in the first control input $\u_1$, given the initial state $x_0$ and the controller $\u_1$, we could define a generic function $J\!:\!\X^{N+1}\times\U^N\!\to\!\RR$ to evaluate the cost incurred by $\u_1$. Therefore, we have the following optimization problem:
\begin{subequations}\label{eq-stlpro}
    \begin{align}
	    &\underset{\u_1,\ldots,\u_n\in\U^N}{\text{minimize}} ~~~~ J(x_0,\u_1 )\\
	     \text{s.t. } &~~~ \texttt{constr}(\phi)\\
      &~~~ x_{t+1}^i=f(x_t^i,u_t^i),\forall t\!=\!0,\ldots,N-1
\end{align}
\end{subequations}
where $x_0^i=x_0$ for any $i=1,\ldots,n$ and $(x_t^1,x_t^2,\ldots,x_t^n)$ is the augmented system state at time instant $t$.
For the sake of simplicity, we denote  the solution of the above optimization problem, 
when $\u_i\in \UU_i$, as a procedure
\[
\u_1,\u_2,\ldots,\u_n\gets \texttt{SolvePlan}(\phi,\UU_1,\UU_2,\ldots,\UU_n).
\]
Hence, if $\texttt{SolvePlan}(\phi,\U^N,\U^N,\ldots,\U^N)$ has a feasible solution, 
then the first component $\u_1$ is the solution to our problem. 
For convenience, in what follows, we use  notation
\[
\xi_f(x_0,\u_1,\ldots,\u_n):=(\xi_f(x_0,\u_1),\ldots,\xi_f(x_0,\u_n))
\]
to represent the trace tuple   generated by each $\u_i$ from   $x_0$
and we also denote by $\xi_f(x_0,\u_1,\ldots,\u_n)\models\phi$ if the STL constraint and system constraint in \eqref{eq-stlpro} are satisfied.

\subsection{Case of HyperSTL with Alternation Depth One}

\begin{algorithm}[t]
	\caption{Control Synthesis  with Depth One}
	\label{alg-1}
	\KwIn{System $\Sigma$ and HyperSTL $\Phi$ in \eqref{eq-alter1}}
	\KwOut{Control input $\u$}
        $\u_1,\u_2,\ldots,\u_n\!\gets\!\texttt{SolvePlan}(\phi,\U^N\!,\U^N\!,\ldots,\U^N)$;\\
        $\hat{\U}_i\gets\{\u_i\},\forall m+1\leq i\leq n$;\\
	\While{True}{
        {\small $\res,\u_{m+1}',\dots,\u_n'\!\!\gets\!\texttt{CountCheck}(\u_1,\ldots,\u_m,\phi)$;}
            \\
        \eIf{\emph{$\res=\true$}}{
        \textbf{return} $\u_1$;
        }{
        $\hat{\U}_i\gets\hat{\U}_i\cup\{\u_i'\},\forall m+1\leq i\leq n$;
        }
        Find $\u_1,\u_2,\ldots,\u_m\in\U^N$ that
        \begin{align}
            &~~~~~~~~\text{minimize}~~J(x_0,\u_1)\notag\\
            \!\!\!\!\!\!\text{s.t. }
            &
            \forall \u_{i}\!\in\!\hat{\U}_{i}\!:\!\xi_f(x_0,\u_1,\ldots,\u_n)\!\models\!\phi,i\!=\!m+1,\ldots,n\notag
        \end{align}
        \\
        \If{\emph{the above problem has no solution}}{
        \textbf{return} \emph{``task $\Phi$ is infeasible"}
        }

	}

        \vspace{5pt}
        \textbf{procedure} $\texttt{CountCheck}(\u_1,\ldots,\u_m,\phi)$\\ 
        Find $\u_{m+1}',\ldots,\u_n'\in\U^N$ that
        \begin{align}
            &\text{minimize}~\rho^\phi(\xi_f(x_0,\u_1,\ldots,\u_m,\u_{m+1}',\ldots,\u_n'))\notag
        \end{align}\\
        \eIf{$\rho^\phi(\xi_f(x_0,\u_1,\ldots,\u_m,\u_{m+1}',\ldots,\u_n'))>0$}{
         \textbf{return } \true, \Null;
        }{
        \textbf{return } \false, $\u_{m+1}',\ldots,\u_n'$;
        }
\end{algorithm}

In this subsection, we further consider the following HyperSTL with alternation depth one
\begin{equation}\label{eq-alter1}
\Phi=\exists\pi_1.\exists\pi_2.\cdots\exists\pi_m.\forall\pi_{m+1}.\forall\pi_{m+2}.\cdots\forall\pi_n.\phi.
\end{equation} 
The key idea is to frame this problem as an STL reactive synthesis task, where the controller must jointly synthesize the first $m$ traces while ensuring robustness against any possible adversarial behaviors in the remaining $n-m$ traces. This problem can be solved using a \emph{counterexample-guided synthesis} approach, as implemented in Algorithm~\ref{alg-1}. Specifically, in line~1, we obtain initial input  $\mathbf{u}_1,\ldots,\mathbf{u}_n \in \mathcal{U}^N$ by calling $\texttt{SolvePlan}$, where the solution space for each input is the entire input space $\mathcal{U}^N$. 
In line~2, for each universally quantified input $\mathbf{u}_i$ (where $m+1 \leq i \leq n$), we initialize candidate domain $\hat{\mathcal{U}}_i$ with the first plan $\mathbf{u}_{i}$. These sets serve as finite constraints for the optimization problem and are updated whenever new counterexamples are encountered.

During the while-loop iteration, the algorithm checks whether the current plan for the first $m$ components satisfies the universality requirement for the last $n\!-\!m$ components using   procedure $\texttt{CountCheck}$.  This procedure attempts to falsify the STL $\phi$ for the fixed current inputs $\mathbf{u}_1,\ldots,\mathbf{u}_m$ by finding counterexample instances $\mathbf{u}_{m+1}',\ldots,\mathbf{u}_n'$. 
If the inputs pass this counterexample check, $\mathbf{u}_1$ constitutes a valid solution. 
Otherwise, the discovered counterexamples are incorporated into the candidate domains $\hat{\mathcal{U}}_{m+1},\ldots,\hat{\mathcal{U}}_n$ in line~8. 
The algorithm then uses these expanded (yet still finite) candidate domains as constraints in line~10 to synthesize a new tuple of inputs $\mathbf{u}_1,\ldots,\mathbf{u}_m$ by solving the optimization problem. If at any point the optimization problem becomes infeasible given the current candidate domains, we can immediately conclude that the entire problem is infeasible.
Note that in procedure $\texttt{CountCheck}$, rather than simply performing quantitative falsification of $\phi$, we additionally minimize the robustness degree of the STL formula as a quantitative optimization objective. This approach tends to yield more effective counterexamples in practice.

\subsection{Control Synthesis for General HyperSTL Formulae}
Now, we are ready to tackle the general case of  HyperSTL synthesis with arbitrary alternation depth of quantifier.

Let $k$ be an index of trace variable. 
For the sake of clarity,  
we can index  the quantifier of each variable in $\Phi$ as follows:
\begin{flalign}\label{eq-general}
    &\Phi=\exists\pi_1. Q_2\pi_2.\cdots Q_{k-1}\pi_{k-1}.(Q_k\pi_k.Q_k\pi_{k+1}.\cdots Q_k\pi_{m_k}.)\notag\\
    &(\bar{Q}_k\pi_{m_k+1}.\bar{Q}_k\pi_{m_k+2}.\cdots \bar{Q}_k\pi_{p_k}.)Q_k\pi_{p_k+1}.\cdots Q_n\pi_n.\phi 
\end{flalign}
where   $\bar{Q}\!=\!\forall$ if $Q\!=\!\exists$ and $\bar{Q}\!=\!\exists$ if $Q\!=\!\forall$.
Intuitively, 
$m_k$ represents the last index of the consecutive quantifiers after $Q_k$ that are same with $Q_k$, i.e., $Q_j\!=\!Q_k,\forall k\leq j\leq m_k$ and $Q_{m_k+1}\!=\!\bar{Q}_k$, while $p_i$ denotes the last index of the consecutive quantifiers after $Q_{m_k+1}$ that are same with $Q_{m_k+1}$, i.e., $Q_j\!=\!\bar{Q}_k,\forall m_k+1\leq j\leq p_k$ and $Q_{p_k+1}=Q_k$.

{\footnotesize
\begin{algorithm}[!ht]
	\caption{Synthesis of General HyperSTL}
	\label{alg-2}
	\KwIn{System $\Sigma$ and general HyperSTL $\Phi$ in \eqref{eq-general}}
	\KwOut{Control input sequence $\u$}
	$\u_1,\ldots,\u_n\gets\texttt{SolvePlan}(\phi,\U^N,\ldots,\U^N)$;\\
        \If{\emph{the above problem has no solution}}{
        \textbf{return } \emph{``task $\Phi$ is infeasible"}\\
        }
        $\hat{\U}_i\gets \{\u_i\},\forall m_1+1\leq i\leq p_1$;\\
        
        \While{\text{True}}
        {
        {\footnotesize $\res, \u_{m_1+1}',\ldots,\u_{p_1}'\gets\texttt{CheckHyper}(\u_1,\ldots,\u_{m_1},\phi)$;}\\
        \eIf{
        \emph{$\res=\true$}
        }{
        \textbf{return} $\u_1$;
        }{
        $\hat{\U}_i\gets\hat{\U}_i\cup\{\u_i'\},\forall m_1+1\leq i\leq p_1$;
        }
        Find $\u_1,\ldots,\u_{m_1}\in\U^N$ that
        \vspace{-4pt}
        \begin{align}
            &~~~~~~~~\text{minimize}~~J(x_0,\u_1)\notag\\
            \!\!\!\!\!\!\text{s.t. }
            &
            \forall \u_{i}\!\in\!\hat{\U}_{i},\exists\u_j\in\U^N:\xi_f(x_0,\u_1,\ldots,\u_n)\!\models\!\phi,\notag\\
            &i=m_1+1,\ldots,p_1,~j=p_1+1,\ldots,n\notag
        \end{align}\\
        \vspace{-4pt}
        \If{\emph{the above problem has no solution}}{
        \textbf{return} \emph{``task $\Phi$ is infeasible"}
        }
        }
        \vspace{2pt}
	\textbf{procedure} $\texttt{CheckHyper}(\u_1,\ldots,\u_{k-1},\psi)$\\
        \eIf{$\psi\!=\!\phi,Q_i\!=\!\forall,\forall k\leq i\leq n$ or $~~~\!\!\psi\!=\!\neg\phi,Q_i\!=\!\exists,\forall k\leq i\leq n$\\}{
        {\small $\res, \w_{k},\ldots,\w_{n}\gets\texttt{CountCheck}(\u_1,\ldots,\u_{k-1},\psi)$;}\\
        \eIf{\emph{$\res=\true$}}{
        \textbf{return} \true, \Null
        }{
        \textbf{return} \false, $\w_{k},\ldots,\w_n$;
        }

        }
        {
        $\w_1,\ldots,\w_{k-1},\w_{k},\w_{k+1},\ldots,\w_n\gets \texttt{SolvePlan}(\psi,\{\u_1\},\ldots,\{\u_{k-1}\},\U^N,\ldots,\U^N)$;\\
        $\hat{\U}_i'\leftarrow\{\w_i\},\forall m_k+1\leq i\leq p_k$;\\
        \While{\emph{\true}}{
        $\res,\w_{m_k+1}',\ldots,\w_{p_k}'\!\!\gets\!\texttt{CheckHyper}(\w_1,\ldots,\w_{k-1},\w_{k},\ldots,\w_{m_k},\neg\psi)$;\\
        \If{\emph{$\res=\true$}}{
        \textbf{return} \false, $\w_{m_k+1},\ldots,\w_{p_k}$;
        }
        $\hat{\U}_i'\gets\hat{\U}_i'\cup\{\w_i'\},\forall m_k+1\leq i\leq p_k$;
        
        Find $\w_{k},\w_{k+1},\ldots,\w_{m_k}\in\U^N$ that
        \vspace{-4pt}
        \begin{align}
            &~~~~~~~~\text{minimize}~~J(x_0,\w_1)\notag\\
            \!\!\!\!\!\!\text{s.t. }
            &
            \forall \w_{i}\!\in\!\hat{\U}_{i},\exists\w_j\in\U^N:\xi_f(x_0,\w_1,\ldots,\w_n)\!\models\!\psi,\notag\\
            &i=m_k+1,\ldots,p_k,~j=p_k+1,\ldots,n\notag
        \end{align} \\
        \vspace{-4pt}
        \If{\emph{the above problem has no solution}}{
        \textbf{return} \emph{True}, \emph{Null};
        }
        
        }
        }
        
\end{algorithm}
}

The key idea for  solving the general case is to recursively call   procedure \texttt{CheckHyper} that is designed according the both procedures \texttt{SolvePlan} and \texttt{CountCheck}. The solution is summarized as Algorithm~\ref{alg-2}. Specifically, in line~1, we obtain initial input $\u_1,\ldots,\u_n\in\U^N$ by \texttt{SolvePlan}. If no feasible input is returned, then we claim that task $\Phi$ is infeasible for system $\Sigma$ in lines 2--3.
Otherwise, we begin to use the \emph{recursive counterexample guided check} to tackle the multiple quantifier alternation in the general HyperSTL $\Phi$ in \eqref{eq-general}. We fix the obtained inputs $\u_1,\ldots,\u_{m_1}$ and view the next $p_1-m_1$ inputs as the check variables. That is, we use procedure \texttt{CheckHyper} to validate whether the traces $\xi_f(x_0,\u_1,\ldots,\u_{m_1})$ satisfy
\begin{equation}
    \forall\pi_{m_1+1}.\cdots\forall\pi_{p_1}.\exists\pi_{p_1+1}.\cdots Q_n\pi_n.\phi.
\end{equation}
Thus, we initiate candidate domain $\hat{\U}_i$ with $\u_{m_1+1},\ldots,\u_{p_1}$ in line 4. 
If satisfied, i.e., $\texttt{CheckHyper}$ returns \emph{True}, we know that the inputs $\u_1,\ldots,\u_{m_1}$ are feasible and return $\u_1$ as the output in lines 7--8. Otherwise, $\texttt{CheckHyper}$ returns \emph{False} together with the counterexample instances $\u_{m_1+1}',\ldots,\u_{p_1}'$. Then, we add them to candidate domains $\hat{\U}_i,\forall i\!=\!m_1+1,\ldots,p_1$ in lines 9--10 and repeat finding new inputs $\u_1,\ldots,\u_n$ in line 11 until $\texttt{CheckHyper}$ returns \emph{True}. If at any point the optimization problem becomes infeasible given the current candidate domains, we immediately claim that the problem is infeasible in lines 12--13.

Next, we elaborate on the details of  procedure $\texttt{CheckHyper}$. 
We denote by $k-1$ the number of its inputs. Then, we   check whether $\xi_f(x_0,\u_1,\ldots,\u_{k-1})$ satisfies \vspace{-2pt}
\begin{equation}\label{eq-remain}
\forall\pi_{k}.\forall\pi_{k+1}.\cdots\forall\pi_{m_k}.\exists\pi_{m_k+1}.\cdots\exists\pi_{p_k}.\forall\pi_{p_{k+1}}.\cdots Q_n\pi_n.\psi.\vspace{-2pt}
\end{equation}
where
$\psi\!\in\!\{\phi,\neg\phi\}$ is the formula to check. 
We first determine whether the remaining HyperSTL is already alternation-free.  If so, i.e., all the remaining quantifiers are $\forall$ or $\exists$, then we try to find a counterexample $\w_k,\ldots,\w_n$ that falsifies $\psi$ by calling \texttt{CountCheck} in line 17. 
If the counterexample does not exist, then we claim that $\xi_f(x_0,\u_1,\ldots,\u_{k-1})$ satisfies  \eqref{eq-remain} and return \emph{True} in lines 18--19. 
Otherwise, we return \emph{False} together with the counterexample in lines 20--21. 
On the other hand, if there are still quantifier alternations, then we fix the given inputs $\u_1,\ldots,\u_{k-1}$ and begin to check whether there are inputs $\w_k,\w_{k+1},\ldots,\w_{m_k}$ that falsify  \eqref{eq-remain}. 
This is equivalent to check whether or not there are inputs $\w_k,\w_{k+1},\ldots,\w_{m_k}$ such that $\xi_f(x_0,\u_1,\ldots,\u_{k-1},\w_k,\ldots,\w_{m_k})$ satisfies
\begin{equation}\label{eq-remainremain}
    \forall\pi_{m_k+1}.\cdots\forall\pi_{p_k}.\exists\pi_{p_k+1}.\cdots \bar{Q}_n\pi_n.\neg\psi.
\end{equation}
To this end, we first find a candidate for $\w_k,\ldots,\w_n$ in line 23. Naturally, we have $\w_i\!=\!\u_i,\forall i\!=\!1,\ldots,k-1$. Then, for the $m_k+1$-th to $p_k$-th inputs, we add $\w_i$ to candidate domains in line 24 and begin to repeat finding $\w_k,\w_{k+1},\ldots,\w_{m_k}$ until the recursive procedure $\texttt{CheckHyper}$ returns \emph{True}, i.e., \eqref{eq-remainremain} is satisfied.
Specifically, if no feasible solution is found, we declare that $\xi_f(x_0,\u_1,\ldots,\u_{k-1})$ satisfies \eqref{eq-remain}, and return \emph{True} in lines 30--33. 
Otherwise, we declare that the inputs $\w_{m_k+1},\ldots,\w_{p_k}$ falsify \eqref{eq-remain} and return \emph{False} in lines 27--29.

Based on the above explanation and analysis for the Algorithm~\ref{alg-2}, we obtain the following result.
\begin{mythm}
    Given   system $\Sigma$ and   general HyperSTL $\Phi$ in \eqref{eq-general}, the control inputs returned by Algorithm~\ref{alg-2} satisfies $\Phi$.
\end{mythm}

\begin{proof}
We finish this proof by induction. For an alternation-free HyperSTL in \eqref{eq-alternationfree}, by lines 1--3, we know that, Algorithm~\ref{alg-2} returns a controller satisfying \eqref{eq-alternationfree} if  there is a feasible solution, or returns \emph{``task $\Phi$ is infeasible"} otherwise. 
For a general HyperSTL in \eqref{eq-general}, we aim to prove that Algorithm~\ref{alg-2} also returns a controller satisfying \eqref{eq-general} if it is feasible and returns \emph{``task $\Phi$ is infeasible"} otherwise. 
In fact, it suffices to prove the correctness of the procedure \texttt{CheckHyper}, i.e., for a given HyperSTL in \eqref{eq-general} and given inputs $\u_1,\ldots,\u_{k-1}$, \texttt{CheckHyper} returns \emph{True} if and only if $\xi_f(x_0,\u_1,\ldots,\u_{k-1})$ satisfies the remaining HyperSTL\vspace{-3pt}
\begin{equation}
    \Psi=Q\pi_k.\cdots Q\pi_{m_k}.\bar{Q}\pi_{m_k+1}.\cdots \bar{Q}\pi_{p_k}.Q\pi_{p_k+1}.\cdots Q_n\pi_n.\psi\vspace{-3pt}
\end{equation}
where $Q_n\!=\!\forall$ if $\psi\!=\!\phi$ and $Q_n\!=\!\exists$ if $\psi\!=\!\neg\phi$. First, it is obvious that \texttt{CheckHyper} is correct for a remaining HyperSTL without quantifier alternation, i.e., \vspace{-3pt}
\begin{equation}\label{eq-free}
    \Psi_0=\forall\pi_k.\cdots\forall\pi_n.\phi\vspace{-3pt}
\end{equation}
since Algorithm~\ref{alg-2} (lines 15--21) is then reduced to Algorithm~\ref{alg-1}. For a given tuple of inputs $\u_1,\ldots,\u_{k-1}$ and a remaining HyperSTL with alternation depth 1, i.e.,
\begin{equation}
    \Psi_1=\exists\pi_k.\cdots\exists\pi_{m_k}.\forall\pi_{m_k+1}.\cdots\forall\pi_n.\phi
\end{equation}
By line 23, we have $\w_i\!=\!\u_i$ for $i\!=\!1,\ldots,k-1$. Then by lines 24--31, we aim to find $\w_k,\ldots,\w_{m_k}$ that \emph{satisfies} $\Psi_1$. To check whether $\w_k,\ldots,\w_{m_k}$ is feasible, it suffices to check whether or not there exists $\w_{m_k+1},\ldots,\w_n$ that \emph{falsifies} $\Psi_1$, i.e., whether or not $\xi_f(x_0,\w_1,\ldots,\w_{m_k})$ \emph{satisfies} the following HyperSTL\vspace{-3pt}
\begin{equation}
    \exists\pi_{m_k+1}.\cdots\exists\pi_n.\neg\phi\vspace{-3pt}
\end{equation}
which is an alternation-free remaining HyperSTL in the form of \eqref{eq-free}. Since \texttt{CheckHyper} is correct for a remaining HyperSTL without quantifier alternation, i.e., the result of $\texttt{CheckHyper}(\w_1,\ldots,\w_{m_k},\neg\phi)$ is correct, then we know that $\texttt{CheckHyper}(\w_1,\ldots,\w_{k-1},\phi)$ will return \emph{True} if and only if there is no $\w_{m_k+1},\ldots,\w_n$ that \emph{falsifies} $\Psi_1$, i.e., $\xi_f(x_0,\u_1,\ldots,\u_{k-1})\models\Psi_1$. Therefore, Algorithm~\ref{alg-2} is also feasible for a remaining HyperSTL with alternation depth 1.

Now, suppose that \texttt{CheckHyper} is correct for a remaining $\Psi$ whose alternation depth is $d$.
Next, we begin to consider a remaining $\Psi$ with alternation depth $d+1$, which is in the following form\vspace{-3pt}
\begin{equation}
    \Psi_{d+1}=\bar{Q}\pi_l.\cdots\bar{Q}\pi_{m_l}.\Psi_d\vspace{-3pt}
\end{equation}
where $l-1$ is the number of given traces.

Without loss of generality, we suppose that $Q\!=\!\forall$ and thus $\bar{Q}\!=\!\exists$.
By line 23, we obtain $\w_i\!=\!\u_i$ for $i\!=\!1,\ldots,l-1$. Similarly, by lines 24--31, we aim to find $\w_l,\ldots,\w_{m_l}$ that \emph{satisfy} $\Psi_{d+1}$. To check whether there exist feasible $\w_l,\ldots,\w_{m_l}$, it suffices to check whether there exist $\w_{m_l+1},\ldots,\w_{p_k}$ that \emph{falsify} $\Psi_{d+1}$, i.e., whether or not $\xi_f(x_0,\w_1,\ldots,\w_{m_l})$ \emph{satisfies} the following HyperSTL\vspace{-3pt}
\begin{equation}
    \exists\pi_{m_l+1}.\cdots\exists\pi_{p_l}.\forall\pi_{p_l+1}.\cdots Q_n\pi_n.\psi\vspace{-3pt}
\end{equation}
whose alternation depth is $d$. Since we supposed that \texttt{CheckHyper} is correct for alternation depth $d$, i.e., the result of $\texttt{CheckHyper}(\w_1,\ldots,\w_{m_l},\neg\psi)$ is correct, then we know that $\texttt{CheckHyper}(\u_1,\ldots,\u_{l-1},\psi)$ will return \emph{True} if and only if there is no $\w_{m_l+1},\ldots,\w_{p_l}$ that \emph{falsify} $\Psi_{d+1}$, i.e., $\xi_f(x_0,\u_1,\ldots,\u_{l-1})\models\Psi_{d+1}$.  

By the above induction, we know that \texttt{CheckHyper} is correct for any HyperSTL. Therefore, by lines 5--13, we know that Algorithm~\ref{alg-2} returns a controller satisfying \eqref{eq-general} if it is feasible and returns \emph{``task $\Phi$ is infeasible"} otherwise. The proof is thus completed.
\end{proof}

Similar to the while loop in Algorithm~\ref{alg-1}, if $\U$ is an infinite set, we still need to set a maximum number of iterations for the two while loops in Algorithm~\ref{alg-2}, whose practicality is hold in the same way in the implementation. As a matter of fact, Algorithm~\ref{alg-2} subsumes Algorithm \ref{alg-1}.

\begin{myrem}
    The HyperSTL controller synthesis problem is solved by a recursive counterexample-guided synthesis technique summarized in Algorithm~\ref{alg-2}. In fact, this technique can also be extended to the control synthesis for other types of hyperproperties specified by real-value constraints for continuous dynamic systems with a slight modification of the optimization problem.
\end{myrem}

\section{Applications of HyperSTL Planning}\label{sec:case}

In this section, we present case studies of the proposed  $\exists$-HyperSTL planning problem by examining it from two perspectives. 
First, we consider the security-preserving planning problem, where a robot must prevent external observers from inferring its critical information.  This problem has been explored extensively in the recent literature  but in purely logical frameworks \cite{shi2023synthesis,zheng2023optimal,he2024security}. Second, we investigate the informed planning problem, in which the robot maintains trajectory indistinguishability from other agents' viewpoints to facilitate collaborative objectives.
\subsection{Security-Preserving Planning}
Following the motivating example, let us consider a mobile robot governed by the dynamics in \eqref{eq-dym}, subject to physical constraints $\X$  and $\U$, tasked with satisfying an STL specification $\phi$. 
We assume the presence of an intruder who possesses complete knowledge of the robot's dynamics and its task specification $\phi$, but lacks access to the robot's exact control strategy. 
Consequently, the intruder must infer the robot's intentions solely through trajectory observations.

From the robot's perspective, certain critical states must be protected, represented by a \emph{secret region} $\X_S \subseteq \X$. The security objective requires that whenever the robot enters $\X_S$, the intruder cannot definitively predict this entry before a specified prediction horizon $\Delta T$ has elapsed. This requirement can be formally captured through the notion of \emph{pre-opacity}  defined as follows:

\begin{mydef}[Pre-Opacity]
Given system $\Sigma$ and STL formula $\phi$, we say the system is \emph{$\Delta T$-step opaque} w.r.t. $\X_S$ if $\T_\Sigma$ satisfies
\begin{equation}
    \exists\pi_1.\exists\pi_2.\!\left[\!\phi^{\pi_1}\wedge\phi^{\pi_2}\wedge\always_{[0,T_\phi]}\!\!\left(
    \begin{aligned}
    &\eventually_{[\Delta T,\Delta T]}(x^{\pi_1}\in\X_S)\\
    & \to[ (x^{\pi_1}=x^{\pi_2})\wedge \\
    & \always_{[0,\Delta T]}\left(x^{\pi_2}\notin \X_S\!\right)]
    \end{aligned}
    \!\right)
    \!\right]\notag
\end{equation}
where $T_{\phi}$ is the evaluation horizon of $\phi$ which is the maximum sum of all nested upper bounds.
\end{mydef}

Intuitively, this definition states that we need to find a trajectory $\pi_1$ satisfying the STL task $\phi$ and there exists at least one trajectory $\pi_2$ satisfying $\phi$ such that if $\pi_1$ reaches the secret region within $\Delta T$ steps, then $\pi_2$ must not reach the secret region within $\Delta T$ steps. Therefore, $\pi_2$ plays as a plausible trajectory such that the intruder cannot determine  that the robot will reach the secret region within $\Delta T$ steps.

As a concrete case study, we consider a mobile robot described by the following unicycle-type nonholonomic model:
\begin{equation}\label{eq-unicycleoriginal}
    \dot{p}_x=v\cos\theta,~\dot{p}_y=v\sin\theta,~\dot{\theta}=\omega,
\end{equation}
where the physical constraints are given by $x\!=\![p_x,p_y,\theta]^T\!\in\!\X\!=\![0,10]^2\!\times\! [-\pi,\pi]$, $[v,\omega]^T\!\in\!\U\!=\! [0,1.0]\!\times\! [-\pi/15,\pi/15]$.
For convenience, we denote that $\p:=[p_x,p_y]^T$. To implement the proposed HyperSTL planning algorithm, we discretize system \eqref{eq-unicycleoriginal} by a sampling time of $0.5\text{s}$. 

The robot operates in the workspace shown in Figure~\ref{fig:opacity} 
and the objective   is  to collect two kinds of packages and then transport them to the final destinations. Formally, we consider the following STL formula:\vspace{-3pt}
    \begin{flalign}\label{eq-stltask}
        \phi=&~\eventually_{[9, 11]}(\p\in A_1)\wedge\always_{[22, 23]}(\p\in A_2)~\wedge\notag\\
        &~\eventually_{[40, 41]}(\p\in (A_3\cup A_4 \cup A_5))\vspace{-3pt}
    \end{flalign}
    where $A_1\!=\![0,2]\!\times\![8,10]$, $A_2\!=\![4,6]\!\times\![4,6]$,  are two warehouses for two kinds of packages, respectively, and $A_3\!=\![8,10]\!\times\![8,10]$, $A_4\!=\![8,10]\!\times\![2.5,4.5]$, $A_5\!=\![8,10]\!\times\![0,2]$ are three available destinations. To improve the performance of the robot trajectory, we define the cost function as\vspace{-3pt}
    \begin{equation}\label{eq-costfunction}
        J(x_0,\u_1)=-\alpha\rho^{\phi}(\xi_f(x_0,\u_1))+(1-\alpha)\sum_{i=1}^{T}\notag(\|\p_i-\p_{i-1}\|^2),\vspace{-3pt}
    \end{equation}
    where $\alpha$ is the trade-off between   the  robustness of the task 
    and the  length of the trajectory. We set $\alpha=0.2$.

 \begin{figure*}
     \centering
     \subfigure[Example for pre-opacity]
     {\includegraphics[scale=0.42]{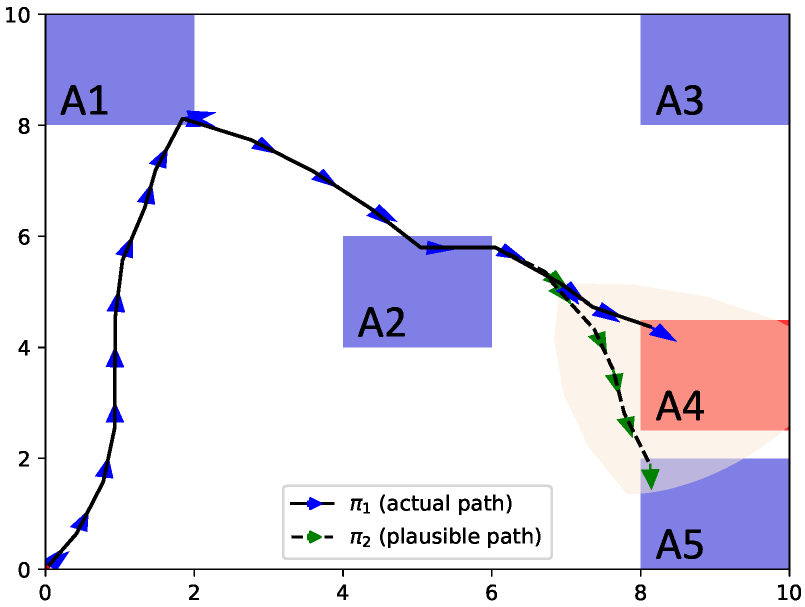}\label{fig:opacity}
     }
     \centering
     \subfigure[Example for anonymity]
     {\includegraphics[scale=0.42]{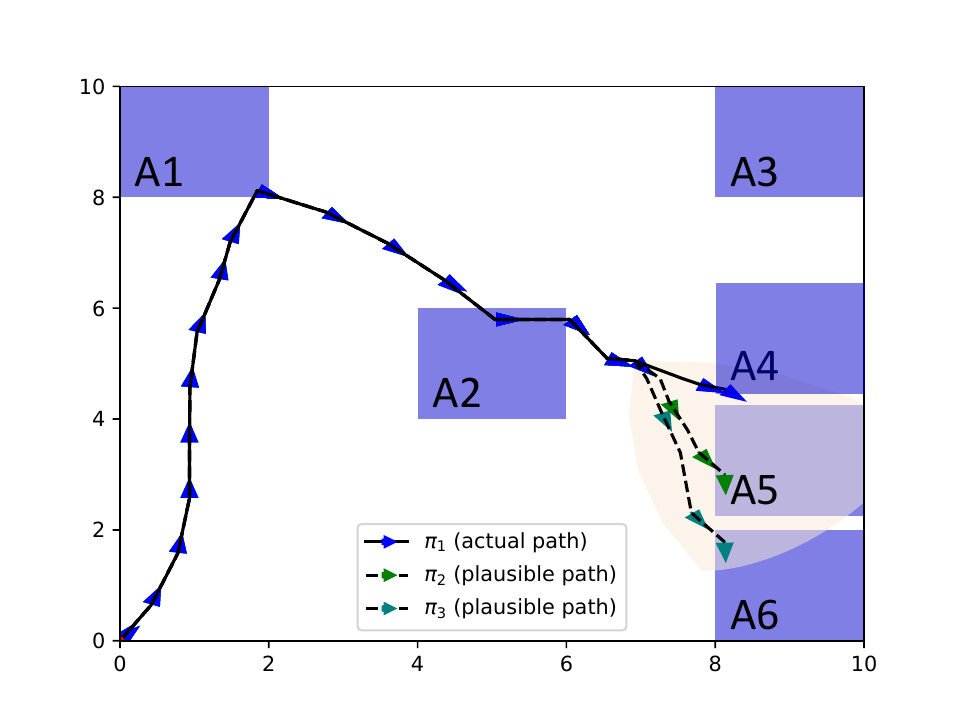}\label{fig:anonymity}
     }
     \centering
     \subfigure[Example for distinguishability]
     {\includegraphics[scale=0.42]{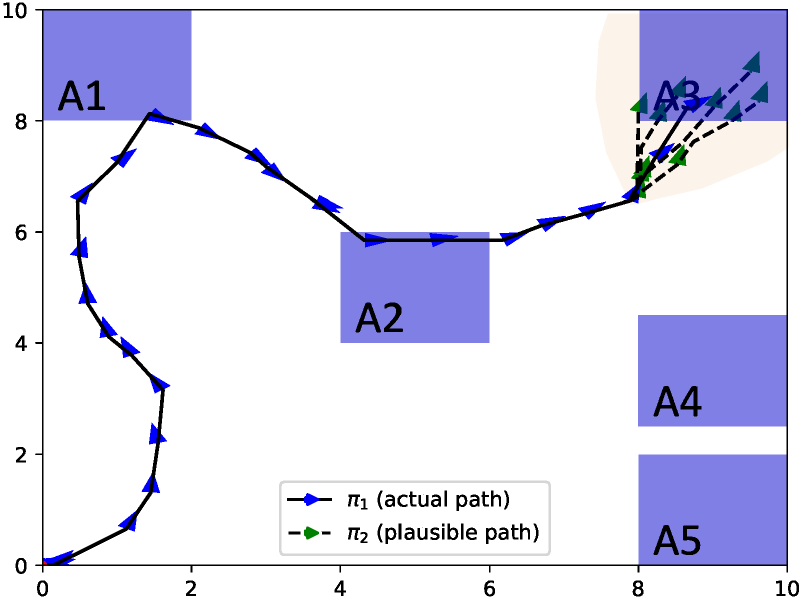}\label{fig:distinguishability}
     }
     \caption{Case studies for HyperSTL planning problems.}
 \end{figure*}

Now, we assume that destination $A_4$ serves as the secret region ($\X_S = A_4$) with an intruder prediction horizon of $\Delta T = 4$. 
We  synthesize a secure control plan $\mathbf{u}_1$ based on the system dynamics and the HyperSTL formulation of pre-opacity. Figure~\ref{fig:opacity} illustrates both the generated trajectory $\pi_1 = \xi_f(x_0, \mathbf{u}_1)$ and its companion trajectory that collectively satisfy the pre-opacity condition. 
The solution trajectory $\pi_1$ successfully accomplishes the STL task $\phi$ from~\eqref{eq-stltask}. 
Moreover, $\Delta$-steps before the robot enters $\X_S$, there still exists an alternative path $\pi_2$ that avoids $\X_S$ yet still satisfies the STL task $\phi$ (by navigating to $A_5$). Thus, even though the intruder observes the robot's exact position at every time step, it cannot conclusively determine whether the robot will enter the secret region $\Delta$-steps in advance when the robot follows the trajectory $\pi_1$.   


In the notion of pre-opacity, it is required that whenever the robot enters the secret region, there must exist an alternative trajectory leading to a non-secret region. However, in some security problems, there is no predefined partition of secret and non-secret regions. Instead, the security requirement is that the robot must maintain at least $K$ plausible alternative intentions. This requirement aligns with the concept of $K$-anonymity in security literature.  

Formally, we assume there exists a set of disjoint regions $\X_A = \X_1 \dot{\cup} \cdots \dot{\cup} \X_n$, and the robot aims to keep its destination within $\X_A$ indistinguishable to an intruder. Specifically, if the robot is predicted to enter one of these regions within a given time horizon, the intruder should remain uncertain about which specific region will be visited. We formalize this requirement as follows.

\begin{mydef}[Anonymity]
    Given  system $\Sigma$ and  STL formula $\phi$, we say the system is \emph{$K$-anonymous} w.r.t. $\X_A$ if $\T_\Sigma$ satisfies
    \begin{align}
        &  \exists\pi_1.\exists\pi_2.\cdots \exists\pi_K. \nonumber \\
        & \always_{[0,T_\phi]}\left(
    \begin{aligned}
        &\eventually_{[\Delta T,\Delta T]}(x^{\pi_1}\!\in\!\X_A)\!\to\![(x^{\pi_1}\!=\!\cdots\!=\!x^{\pi_K})\wedge\\
        &
        \bigvee_{\{m_1,\dots, m_K\}\in I_K} \bigwedge_{i=1,\dots, K} \eventually_{[0,\Delta T]}(x^{\pi_i}\in\X_{m_i}) 
        ]
    \end{aligned}
    \right)\nonumber
    \end{align}
{where $I_K\subseteq 2^{\{1,\dots,n\}}$ is the set of all index sets with cardinality $K$.}
\end{mydef}

As an illustrative example, let us still consider the nonholonomic mobile robot \eqref{eq-unicycleoriginal} with the same physical constraints $\X$ and $\U$. 
The STL task $\phi$ is modified to \vspace{-3pt}
    \begin{flalign}\label{eq-stltask}
        \phi=\eventually_{[9, 11]}(\p\!\in\! A_1)\wedge\always_{[22, 23]}(\p\!\in\! A_2) \wedge 
         \eventually_{[40, 41]}(\p\!\in\! \X_A),\notag\vspace{-3pt}
    \end{flalign}
where $\X_A\!=\!A_3\dot{\cup}A_4\dot{\cup}A_5\dot{\cup}A_6$ as shown in Figure~\ref{fig:anonymity}. 
Here the critical regions are the four  possible destinations and 
we set \( K = 3 \) for anonymity. 
Using the system dynamics in~\eqref{eq-unicycleoriginal} and the HyperSTL specification for anonymity,
we compute  control input  \(\mathbf{u}_1\). 
The resulting trajectory is illustrated in Figure~\ref{fig:anonymity}.
When the robot follows path \(\pi_1\), it maintains three plausible destinations before reaching the actual target \(A_4\): \(A_4\) itself along with alternatives \(A_5\) and \(A_6\). This satisfies the \(3\)-anonymity requirement by preserving ambiguity about the final destination.

\subsection{Informed Planning without Ambiguity}
While the previous case studies addressed security concerns, $\exists$-HyperSTL can also model collaborative scenarios under implicit communication. 
We still consider the robot planning setting. 
Rather than an intruder observing the trajectory, we now assume a UAV monitors the robot's state in real-time for cooperative purposes. 
Without direct communication, the robot must ensure its plan remains unambiguous to prevent misinterpretation by the UAV.

Specifically, we define a critical region $\X_C\subseteq\X$ that the robot must visit. 
The UAV must be able to predict each visit to $\X_C$ with sufficient advance notice (a certain number of steps) to prepare appropriate assistance. 
This requirement can be formally characterized as follows.

\begin{mydef}[Distinguishability]
Given system $\Sigma$ and STL formula $\phi$, we say the system is \emph{$\Delta T$-step distinguishable} w.r.t. $\X_C$ if $\T_\Sigma$ satisfies
\begin{equation}
    \exists\pi_1.\forall\pi_2.\!\!\left[\phi^{\pi_1}
    \!\wedge\!\!\left(
    \phi^{\pi_2}\!\to\!
    \always_{[0,T_\phi]}\left(
    \begin{aligned}
        &\eventually_{[\Delta T,\Delta T]}(x^{\pi_1}\in\X_C)\\
        &\to[(x^{\pi_1}=x^{\pi_2})\to\\
        &\eventually_{[0,\Delta T]}(x^{\pi_2}\in\X_C)]
    \end{aligned}
    \right)
    \!\right)
    \!\right]\notag
\end{equation}
\end{mydef}

Intuitively, this definition states that we need to find a trajectory $\pi_1$ satisfying the STL task $\phi$ such that for any other trajectory $\pi_2$ also satisfying $\phi$, if $\pi_2$ reaches the critical region within $\Delta T$ steps, then $\pi_1$ must also reach the critical region within $\Delta T$ steps.  
Thus, there is no ambiguity about when the critical region must be visited.

As a case study, we still consider the nonholonomic mobile robot in \eqref{eq-unicycleoriginal} operating in the workspace shown in Figure~\ref{fig:distinguishability}. 
The STL task $\phi$ remains the same as in \eqref{eq-stltask}, and we use the cost function defined in \eqref{eq-costfunction}.
Let the critical region be $\X_C = A_3$, and set the prediction horizon to $\Delta T = 2$.  
Using Algorithm~\ref{alg-2}, we synthesize an unambiguous control plan $\mathbf{u}_1$ based on the system dynamics  and the HyperSTL formulation of distinguishability. 
The resulting trajectory $\pi_1 = \xi_f(x_0, \mathbf{u}_1)$ is depicted as the blue line in Figure~\ref{fig:distinguishability}.  
Observe that all other plausible trajectories satisfying $\phi$ (shown as green lines) also reach the critical region within two steps. 
This ensures that the UAV can uniquely determine the robot’s intended destination, enabling appropriate assistance actions to be taken.  

\vspace{-3pt}
\section{Conclusions}\label{sec:conclusion}
In this paper, we investigate the task and motion planning for dynamical systems under signal temporal logic  specifications. 
Unlike existing STL control synthesis approaches that impose temporal properties solely on individual trajectories, we leverage HyperSTL semantics to characterize inter-relationships between multiple system executions. 
Our approach  integrates  mixed-integer programming  optimization and counterexample-guided synthesis techniques in a novel manner to solve the control synthesis problem. 
Through case studies involving information-flow synthesis for dynamical systems, we demonstrate the effectiveness of our approach. 
Note that while our work focuses on planning problems, 
it is also applicable to HyperSTL model checking for discrete-time dynamical systems. In the future, we plan to extend our results to stochastic settings to achieve HyperSTL planning with probabilistic guarantees.

\bibliographystyle{plain}
\bibliography{myref}

\end{document}